\documentclass[a4paper,UKenglish,cleveref, autoref, thm-restate]{lipics-v2021}
\usepackage{microtype} 

\hideLIPIcs 

\usepackage{graphicx} 
\usepackage{hyperref}
\usepackage{amsmath,amssymb,yhmath}
\usepackage[usenames,dvipsnames]{xcolor}
\usepackage{xspace,comment}
\usepackage{algorithm}
\usepackage{mathtools}
\usepackage[noend]{algpseudocode}

\newcommand{\arc}{\wideparen}

\newcommand{\Minsumradius}{\textsc{MinSumRadius}\xspace}
\newcommand{\Minsumdiameter}{\textsc{MinSumDiameter}\xspace}
\newcommand{\Center}{\textsc{Center}\xspace}
\newcommand{\Means}{\textsc{Means}\xspace}
\newcommand{\Oh}{\ensuremath{\mathcal{O}}}

\newcommand{\RR}{\ensuremath{\mathbb{R}}}
\newcommand{\med}{\mathrm{MED}}
\newcommand{\reals}{\mathbb{R}}
\graphicspath{{figures/}}
\newcommand{\MEB}{\ensuremath{\mathrm{MEB}}}

\title{Clustering with Few Disks to Minimize the Sum~of~Radii}

\author{Mikkel Abrahamsen}{University of Copenhagen, Denmark}{miab@di.ku.dk}{}{is supported by Starting Grant 1054-00032B from the Independent Research Fund Denmark under the Sapere Aude research career programme and is part of Basic Algorithms Research Copenhagen (BARC), supported by the VILLUM Foundation grant 16582.}
\author{Sarita de Berg}{Utrecht University, The Netherlands}{S.deBerg@uu.nl}{}{}
\author{Lucas Meijer}{Utrecht University, The Netherlands}{l.meijer2@uu.nl}{}{is supported by the Netherlands Organisation for Scientific Research (NWO) under project no. VI.Vidi.213.150.}
\author{André Nusser}{University of Copenhagen, Denmark}{annu@di.ku.dk}{}{is part of Basic Algorithms Research Copenhagen (BARC), supported by the VILLUM Foundation grant 16582.}
\author{Leonidas Theocharous}{Eindhoven University of Technology, The Netherlands}{l.theocharous@tue.nl}{}{ is supported by the Dutch Research Council (NWO) through Gravitation-grant NETWORKS-024.002.003.}

\authorrunning{M. Abrahamsen, S. de Berg, L. Meijer, A. Nusser, L. Theocharous}

\Copyright{Mikkel Abrahamsen, Sarita de Berg, Lucas Meijer, André Nusser, Leonidas Theocharous}



\EventEditors{Wolfgang Mulzer and Jeff M. Phillips}
\EventNoEds{2}
\EventLongTitle{40th International Symposium on Computational Geometry
(SoCG 2024)}
\EventShortTitle{SoCG 2024}
\EventAcronym{SoCG}
\EventYear{2024}
\EventDate{June 11-14, 2024}
\EventLocation{Athens, Greece}
\EventLogo{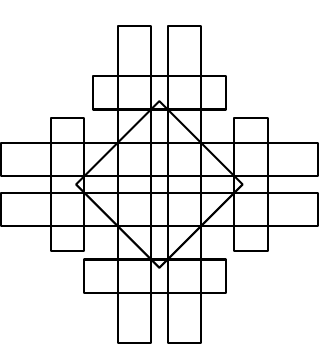}
\SeriesVolume{293}
\ArticleNo{XX}     

\category{} 

\relatedversion{}\relatedversiondetails{Full Version}{https://arxiv.org/abs/2312.08803} 

\funding{}

\acknowledgements{This work was initiated at the Workshop on New Directions in Geometric Algorithms, May 14--19 2023,
Utrecht, The Netherlands.
}

\nolinenumbers 

\begin{document}

\maketitle

\begin{abstract}
Given a set of $n$ points in the Euclidean plane, the $k$-\textsc{MinSumRadius} problem asks to cover this point set using $k$ disks with the objective of minimizing the sum of the radii of the disks.
After a long line of research on related problems, it was finally discovered that this problem admits a polynomial time algorithm [GKKPV~'12]; however, the running time of this algorithm is $\mathcal{O}(n^{881})$, and its relevance is thereby mostly of theoretical nature.
A practically and structurally interesting special case of the $k$-\textsc{MinSumRadius} problem is that of small $k$.
For the $2$-\textsc{MinSumRadius} problem, a near-quadratic time algorithm with expected running time $\mathcal{O}(n^2 \log^2 n \log^2 \log n)$ was given over 30 years ago [Eppstein~'92].

We present the first improvement of this result, namely, a near-linear time algorithm to compute the $2$-\textsc{MinSumRadius} that runs in expected $\mathcal{O}(n \log^2 n \log^2 \log n)$ time. We generalize this result to any constant dimension $d$, for which we give an $\mathcal{O}(n^{2-1/(\lceil d/2\rceil + 1) + \varepsilon})$ time algorithm.
Additionally, we give a near-quadratic time algorithm for $3$-\textsc{MinSumRadius} in the plane that runs in expected $\mathcal{O}(n^2 \log^2 n \log^2 \log n)$ time.
All of these algorithms rely on insights that uncover a surprisingly simple structure of optimal solutions: we can specify a linear number of lines out of which one separates one of the clusters from the remaining clusters in an optimal solution. 
\end{abstract}

\section{Introduction}

Clustering seeks to partition a data set in order to obtain a deeper understanding of its structure.
There are different clustering notions that cater to different applications.
An important subclass is geometric clusterings \cite{geometric_clusterings}.
In their general form, as defined in \cite{geometric_clusterings}, geometric clusterings try to partition a set of input points in the plane into $k$ clusters such that some objective function is minimized.
More formally, let $f$ be a symmetric $k$-ary function and $w$ a non-negative function over all subsets of input points.
The geometric clustering problem is defined as follows:
Given a set $P$ of points in the Euclidean plane and an integer $k$, partition $P$ into $k$ sets $C_1, \dots, C_k$ such that $f(w(C_1), \dots, w(C_k))$ is minimized.
Popular choices for the weight function $w$ are the radius of the minimum enclosing disk and the sum of squared distances from the points to the mean.
The function $f$ aggregates the weights of all clusters, for example using the maximum or the sum.

Arguably the most popular types of clustering in this setting are $k$-\Center clustering (with $f$ being the maximum, and $w$ being the radius of the minimum enclosing disk) and $k$-\Means clustering (with $f$ being the sum, and $w$ being the sum of squared distances to the mean).
While geometric clusterings have the advantage that their underlying objective function can be very intuitive, unfortunately the cluster boundaries might sometimes be slightly more complex.
For example, the disks defining the clusters can have a large overlap in $k$-\Center clustering, while in $k$-\Means clustering the boundaries are defined by the Voronoi diagram on the mean points of the clusters (whose complexity has an exponential dependency on the dimension).
An instance of geometric clustering for which the cluster boundaries implicitly consist of non-overlapping disks is $k$-\Minsumradius.
In $k$-\Minsumradius clustering we want to minimize the sum of radii of the $k$ disks with which we cover the input point set.
In the geometric clustering setting, this means that the function $f$ is the sum and $w$ is the radius of the minimum enclosing disk.
This is the notion of clustering that we consider in this work.

While $k$-\Center and $k$-\Means are both NP-hard in the Euclidean plane when $k$ is part of the input \cite{kmeans_nphard_plane,kcenter_nphard}, the $k$-\Minsumradius problem can be solved in polynomial time \cite{siam_jc}. Interestingly, this result extends to any Euclidean space when its dimension is considered a constant.
However, the bound on the running time
\footnote{As the radii of minimum enclosing disks can contain square roots, the value of a solution is a sum of square roots. However, it is not known how to compare two sums of square roots in polynomial time in the number of elements. The running time merely counts the number of such comparisons. Additionally, the setting is slightly different than ours in that it assumes that the centers of the clusters have to be located at input points. We believe that this algorithm also carries over to the setting of unrestricted cluster centers, incurring some additional polynomial factors in the running time.} is $\Oh(n^{881})$.
Whether there exists a practically fast exact algorithm for $k$-\Minsumradius can be considered one of the most interesting questions in geometric clustering.
While the running time of the known polynomial-time algorithm can likely be slightly improved using the same techniques, the balanced separators that are used in the algorithm inevitably lead to a high running time.
Thus, we believe that further structural insights into the problem --- especially with respect to separators --- are needed to greatly reduce the exponent of the polynomial running time.

In order to obtain a deeper understanding of the problem structure and due to its practical relevance, in this work we consider the $k$-\Minsumradius problem for small values of $k$.
The importance of this setting is reflected in the extensive work that was conducted in the analogous setting for the $k$-\Center and $k$-\Means problem, especially for the case of two clusters \cite{2center_quadratic,2means_nphard1,2center_linear3,2center_linear5,2means_nphard2,2center_linear2,kmeans_fpt,2center_linear1,2center_linear4} --- also called \emph{bi-partition}.
Bi-partition problems are of interest on their own, but they can additionally be used as a subroutine in hierarchical clustering methods.
Even more interestingly, while near-linear time algorithms for $2$-\Center clustering received a lot of attention \cite{2center_linear3,2center_linear2,2center_linear5,2center_linear1,2center_linear4}, the best known algorithm for $2$-\Minsumradius still has near-quadratic expected running time $\Oh(n^2 \log^2 n \log^2 \log n)$ \cite{k2quadratic}, which has seen no improvement in the last 30 years, despite significant work on related problems.

We are aware of no reason why $2$-\Center should deserve more attention than $2$-\Minsumradius, and we believe that the preference for the former problem may have arisen merely out of habit or historical trends in research focus.
The idea that similar data points are grouped together in the same cluster lies at the heart of clustering and, in particular, this is the reason why clustering is used for classification tasks in machine learning.
We note that in this regard, as shown in Figure \ref{fig:5}, optimal solutions to $2$-\Minsumradius often result in more natural clusters than for $2$-\Center.

\begin{figure}
\centering
\includegraphics[page=5]{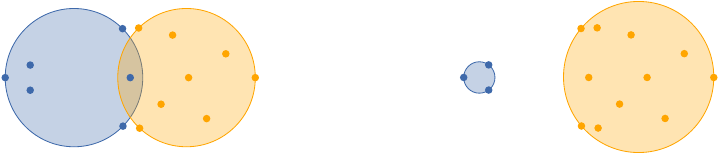}
\caption{The optimal $2$-\Center clustering (left) compared to the optimal $2$-\Minsumradius clustering (right) for the same point set. In this example $2$-\Minsumradius clustering better captures the structure of the point set than $2$-\Center clustering.}
\label{fig:5}
\end{figure}

\paragraph*{Our results}
In this work, we break the quadratic barrier for $2$-\Minsumradius in the Euclidean plane by presenting a near-linear time algorithm with expected running time $\Oh(n \log^2 n \log^2 \log n)$. Our method actually extends to any constant dimension $d\geq 3$, again yielding a subquadratic algorithm with running time $\mathcal{O}(n^{2-1/(\lceil d/2\rceil + 1) + \varepsilon})$. 
Moreover, we extend our structural insights to planar $3$-\Minsumradius\ --- matching the previously best $2$-\Minsumradius running time --- and give an algorithm with expected $\Oh(n^2 \log^2 n \log^2 \log n)$ running time, which is the first non-trivial result on this special case that we are aware of. The running time for planar $2$-\Minsumradius and $3$-\Minsumradius can be made deterministic by using the deterministic algorithm to maintain a minimum enclosing ball~\cite{k2quadratic}, which increases the running times to $\Oh(n \log^4n)$ and $\Oh(n^2\log^4 n)$, respectively.

The main technical contribution leading to these results are structural insights that simplify the problem significantly.
Concretely, we show that the points on the boundary of the minimum enclosing disk (or ball, in higher dimensions) of the point set induce a constant number of directions such that there is a line (or hyperplane) with one of those directions that separates one cluster from the other clusters in an optimal solution.
As there are only linearly many combinatorially distinct separator lines for each direction, we have linearly many separators in total that we have to consider. Note that this is the main difference to the previously best algorithm for planar $2$-\Minsumradius~\cite{k2quadratic}, which considered quadratically many separators.
We then check all clusterings induced by these separators using an algorithm from \cite{k2quadratic} to dynamically maintain a minimum enclosing disk and, in the $k=3$ case, we use our $2$-\Minsumradius algorithm as subroutine. For the higher-dimensional $2$-\Minsumradius problem, we similarly use an algorithm to maintain a minimum enclosing ball in any dimension $d\geq 3$ \cite{dynamic91}. 
While our algorithms are interesting in their own right, we additionally hope to enable a better understanding of the general case by uncovering this surprisingly simple structure of separators in the cases $k\in\{2,3\}$.

In Section~\ref{sec:k2}, we discuss the near-linear time algorithm for $2$-\Minsumradius, including the sub-quadratic time algorithm for $d \geq 3$.
And in Section~\ref{sec:k3} we discuss the near-quadratic time algorithm for $3$-\Minsumradius.
We conclude our work in Section~\ref{sec:conclusion}.

\paragraph*{Other related work}
There is a multitude of work on more general or similar problems. In a general metric space, constant-factor approximation algorithms for $k$-\Minsumradius have been known for a long time \cite{CharikarP01}. Recently, a $3$-approximation algorithm was presented, which also works in the setting where outliers are allowed \cite{DBLP:journals/corr/abs-2311-06111}.

In \cite{geometric_clusterings}, the authors consider the problem of clustering a set of points into $k$ clusters with an objective function that is monotonically increasing in the diameters or radii of the clusters.
Concretely, they give an algorithm with running time $\Oh(n^{6k})$.
In the same work, the authors consider separators of these types of clusterings.
However, as opposed to our work, the authors rather focus on the existence of separators in their general setting, while we try to reduce the number of relevant separators to consider in $k$-\Minsumradius.

In \cite{alpha}, a generalized version of $k$-\Minsumradius is studied, namely, that of minimizing the sum of radii taken to the power of $\alpha$.
The authors show NP-hardness for $\alpha > 1$ in this setting, but they give a polynomial-time algorithm for the special case in which all centers of disks have to lie on a line.

Another similar setting is that of minimizing the sum of diameters of the clusters, which we call $k$-\Minsumdiameter.
We want to emphasize that when referring to the radius of a cluster, this refers to the radius of the minimum enclosing disk, while when referring to the diameter, it commonly refers to the distance between the furthest pair of points in the cluster.
In contrast to the case of $k$-\Minsumradius, in $k$-\Minsumdiameter the cluster boundaries do not behave as nicely.
For the $2$-\Minsumdiameter problem, a near-linear time algorithm in the Euclidean plane was given more than 30 years ago~\cite{DBLP:journals/comgeo/Hershberger92a}.
However, for this problem, different techniques seem to be relevant, as farthest-point Voronoi diagrams are utilized rather than separators.
A similar problem to $k$-\Minsumdiameter is the problem of minimizing the maximum diameter of the clusters, which has a near-quadratic time algorithm for the case of three clusters~\cite{DBLP:journals/comgeo/HagauerR97}. 

Opposed to enclosing points using disks, there is a class of problems that aims at enclosing a set of points with multiple closed polylines while minimizing the sum of the perimeters. This problem and similar variants were studied in the general setting of arbitrarily many clusters~\cite{fast_fencing} as well as in the setting of two clusters~\cite{DBLP:journals/dcg/AbrahamsenBBMM20}.

Finally, very recently there was also some work to determine the fine-grained running times of the $k$-\Center problem --- partly for other shapes than disks --- for small values of~$k$~\cite{DBLP:conf/icalp/ChanHY23}.

\section{Preliminaries}
For a set of points $Q\subset \reals^d$, let $\Call{MEB}{Q}$ be the minimum enclosing ball that contains all points of $Q$.
We denote by $r(Q)$ the radius of $\Call{MEB}{Q}$. For $d=2$, we denote the minimum enclosing ball, which is a disk, by $\Call{MED}{Q}$.

\begin{definition}[$k$-\Minsumradius]
Let $P$ be a set of $n$ points in $\reals^d$ and $k$ be a positive integer. The $k$-\Minsumradius problem asks to partition $P$ into $k$ clusters $C_1,C_2,...,C_k$ such that $\sum_{i=1}^k r(C_i)$ is minimized.
\end{definition}
Throughout the paper, we let $D$ denote the minimum enclosing ball of the input set of our $k$-\Minsumradius instance, i.e., $D \coloneqq \Call{MEB}{P}$, and let $c$ be the center of $D$.
We say that a point set $Q$ defines a ball $D'$ if $\Call{MEB}{Q} = D'$. 
Given $D$ and a point $p$ on the boundary of this ball, we define $p^*$ to be the diametrically-opposing point of $p$ on $D$.
We define $d(p)$ to be the diameter of $D$ that contains $p$, i.e., $d(p)$ is the segment from $p$ to $p^*$. 

The next lemma states that in an optimal solution clusters are well-separated, in the sense that their minimum enclosing balls are disjoint.
\begin{lemma} \label{non-intersecting}
    Given a $k$-\Minsumradius instance, there exists an optimal solution with clusters $C_1, \dots, C_k$ such that $\Call{MEB}{C_i} \cap \Call{MEB}{C_j} = \emptyset$ for all distinct $i,j \in [k]$.
\end{lemma}
\begin{proof}
Consider an arbitrary optimal solution and assume that there are two clusters $C_i, C_j$ for which $\Call{MEB}{C_i} \cap \Call{MEB}{C_j} \neq \emptyset$.
We can replace the clusters $C_i, C_j$ by a single cluster $C' \coloneqq C_i \cup C_j$. Note that this does not increase the total clustering cost as $r(C') \leq r(C_i) + r(C_j)$.
Hence, we obtain a clustering with one less cluster (or, in other words, one more empty cluster).
We can recursively apply this argument until we either end up with a single cluster (which trivially satisfies the lemma) or all clusters have pairwise non-overlapping minimum enclosing balls.
\end{proof}

We also make use of the following (slightly differently phrased) known lemma, which intuitively states that any plane through the center of a minimum enclosing ball has points on both sides.
\begin{lemma}{\cite[Lemma 2.2]{DBLP:conf/stoc/BadoiuHI02}} \label{lem:no_empty_half}
    Let $Q \subset \RR^d$ with $|Q| = d+1$ be a set of points and let $\MEB(Q)$ be their minimum enclosing ball. There exists no hyperplane $H$ through the center of $\MEB(Q)$ such that all points of $Q$ are strictly on one side of $H$.
\end{lemma}

We say that a $(d-1)$-dimensional hyperplane $\ell$ \emph{separates} two sets $A,B\subset\mathbb R^d$ if $A$ and $B$ are contained in different closed half-spaces defined by $\ell$ and $\ell\cap A=\emptyset$ or $\ell\cap B=\emptyset$.
In our algorithms for 2-\Minsumradius and 3-\Minsumradius in the plane, we consider a sweep-line $\ell$ perpendicular to a direction defined by points on the boundary of $D$, keeping track of the points in $P$ on both sides of $\ell$.
Our definition of a separator is convenient, as it allows us to include all input points on $\ell$ together with the points on one of the sides, thus avoiding any general position assumptions.


\section{Near-linear algorithm for 2-\Minsumradius}
\label{sec:k2}
In this section, we present our algorithm for $2$-\Minsumradius for both the plane and higher dimensions. First, we consider the problem in the plane to build an intuitive understanding, which is also applicable to the $3$-\Minsumradius problem, before generalizing to higher dimensions in Section~\ref{sec:k2_generalize}. For the plane, we prove the following result:
\begin{theorem}\label{thm:k2}
For a set $P$ of $n$ points in the Euclidean plane, an optimal 2-\Minsumradius clustering can be computed in expected $\Oh(n \log^2 n \log^2 \log n)$ or worst-case $\Oh(n\log^4n)$ time.
\end{theorem}
In Section~\ref{sub:algo_k2} we present the algorithm, and in Section~\ref{sec:k2_cuts} we prove the main structural result that our algorithm relies on.

\subsection{Algorithm}
\label{sub:algo_k2}

Our algorithm uses the insight that there exist a linear number of separators, such that one of them separates the points in cluster $C_1$ from those in cluster $C_2$.

\begin{algorithm}[t]
\begin{algorithmic}[1]
\Procedure{2-\Minsumradius}{$P$}
\State $D \gets $ minimum enclosing disk of $P$
\State $p_1, p_2, p_3$ $\gets$ points defining $D$
\State $S_{\text{OPT}} \gets (D, \emptyset)$ \Comment{single cluster containing all of $P$}
\For{each diameter $d(p_i)$ at $p_i$}
    \State initialize dynamically maintained MEDs: $A \gets \Call{MED}{\emptyset}, B \gets \Call{MED}{P}$
	\For{each separator orthogonal to $d(p_i)$ through a $p \in P$ in sorted order}
        \State update $A$ adding $p$, and update $B$ removing $p$
		\If{$(A,B)$ is a better solution than $S_{\text{OPT}}$}
			\State $S_{\text{OPT}} \gets (A,B)$
		\EndIf
	\EndFor
\EndFor
\State \Return $S_{\text{OPT}}$
\EndProcedure
\end{algorithmic}
\caption{Near-linear time algorithm for 2-\Minsumradius.}
\label{alg:k2}
\end{algorithm}

\begin{lemma} \label{lem:k2_separator}
%
Consider a point set $P$ in the Euclidean plane and let $p_1, p_2, p_3$ be three points of $P$ that define the minimum enclosing disk of $P$ (with potentially $p_2 = p_3$).
There exists an optimal 2-\Minsumradius clustering $C_1, C_2$ of $P$ such that for a point $q\in\{p_1,p_2,p_3\}$ and a line $\ell$ orthogonal to $d(q)$, the line $\ell$ separates $C_1$ from $C_2$.
\end{lemma}
We prove this result in Section~\ref{sec:k2_cuts}.
We now explain our algorithm relying on Lemma~\ref{lem:k2_separator}; see Algorithm~\ref{alg:k2} for the pseudo-code.
\subparagraph{Algorithm description.} 
Let $p_1, p_2, p_3$ denote a triple of points in $P$ that define $D$ (with possibly $p_2 = p_3$).
These points can be computed in $\Oh(n)$ time~\cite{DBLP:journals/siamcomp/Megiddo83a}.
We try out every combinatorially distinct line orthogonal to $d(p_i)$ for $i \in \{1, 2, 3\}$ as a separator.
We consider the separators orthogonal to a specific $d(p_i)$ in sorted order such that in each step only one point jumps over the separator.
This ensures that the minimum enclosing disks on both sides of the separator only change by a single point per step, and we can therefore use a data structure to dynamically maintain them.
We then select the best solution found using these separators.

Correctness follows directly from Lemma~\ref{lem:k2_separator}, hence let us consider the running time.
We can maintain the minimum enclosing disks $A$ and $B$ in expected $\Oh(\log^2 n \log^2 \log n)$ or worst-case $\Oh(\log^4 n)$ amortized time per update~\cite{k2quadratic}.
So, checking all $n$ separators requires expected $\Oh(n \log^2 n \log^2 \log n)$ or $\Oh(n \log ^4 n)$ worst-case time.
As we only have three diameters $d(p_1),d(p_2),d(p_3)$ to handle, we can compute the optimal solution in the same time.

\subsection{Linear number of cuts}
\label{sec:k2_cuts}
In this subsection we prove Lemma~\ref{lem:k2_separator}, which states that it suffices to only check separators orthogonal to one of the diameters $d(p_1)$, $d(p_2)$, or $d(p_3)$.

\begin{proof}[Proof of Lemma~\ref{lem:k2_separator}]
Let $C_1,C_2$ be an optimal clustering.
If $r(C_1)+r(C_2)=r(D)$, then the clustering $P,\emptyset$ is also optimal and trivally satisfies the claim.
We therefore in the following assume that $r(C_1)+r(C_2)<r(D)$.
We consider two different cases, depending on whether the minimum enclosing disk $D$ of the input points is defined by two or more points.

\begin{figure}
    \centering
    \begin{minipage}[b]{.45\textwidth}
        \centering
        \includegraphics[page=2]{vertical-separators.pdf}
        \subcaption{Case 1: Two points define $D$.}
        \label{fig:2points_define}
    \end{minipage}
    \begin{minipage}[b]{.45\textwidth}
        \centering
        \includegraphics[page=1]{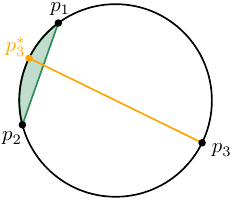}
        \subcaption{Case 2: At least three points define $D$.}
        \label{fig:3points_define}
    \end{minipage}
    \caption{Visualization of the two cases in the proof for the existence of a separator that is perpendicular to one of the diameters of $p_1, p_2, p_3$.}
\end{figure}

\subparagraph{Case 1: } \textit{Two points define $D$.}
We illustrate this case in Figure~\ref{fig:2points_define}.
Let $p_1$ and $p_2$ be the two points that define the minimum enclosing disk.
Then, $p_1$ must be diametrically opposing $p_2$ on $D$, so $d(p_1)=d(p_2)$.
As we assume that the optimal solution has value less than $r(D)$, we have that $p_1$ and $p_2$ must belong to different clusters; without loss of generality, let $p_1\in C_1$ and $p_2\in C_2$. 

Let $\ell$ be the line tangent to $\med(C_1)$ perpendicular to $d(p_1)$ furthest in direction $\overrightarrow{c-p_1}$ (recall that $c$ is the center of $D$).
If $\ell\cap \med(C_2)=\emptyset$, we are done.
Otherwise, the projections of the disks $\med(C_1)$ and $\med(C_2)$ on $d(p_1)$ cover all of the diameter, so $r(C_1)+r(C_2)\geq r(D)$,
which is a contradiction.

\subparagraph{Case 2: } \textit{At least three points define $D$.} \label{k2:c2}
We illustrate this case in Figure~\ref{fig:3points_define}. Let $p_1, p_2, p_3\in P$ be any three points that jointly define the minimum enclosing disk.
In any optimal solution, two out of these three points must be grouped together in a cluster.
Without loss of generality, assume that $p_1$ and $p_2$ are in the same cluster.

Any disk containing $p_1$ and $p_2$ that is defined by points in $D$ must contain all of $\arc{p_1p_2}$, the smallest of the two arcs on $D$ connecting $p_1$ and $p_2$.
We now add the artificial point $p_3^*$ to the point set --- the diametrically opposing point of $p_3$.
We have that $p_3^*\in \arc{p_1p_2}$ as otherwise $p_1, p_2, p_3$ would lie strictly inside one half of $D$ and could therefore not define the minimum enclosing disk by Lemma~\ref{lem:no_empty_half}.
Hence, adding $p_3^*$ to the point set does not change the optimal solution.
Thus, we reduced our problem to Case~1, where two points define the minimum enclosing disk.
\end{proof}

\subsection{Extension to constant dimensions} \label{sec:k2_generalize}

In this section, we show how to extend our algorithm for 2-\Minsumradius to any constant dimension $d$.

Our main structural lemma in this section is an analog of Lemma~\ref{lem:k2_separator} generalized to higher dimensions. From this lemma, it follows that we can essentially brute-force the clusters over a small number ($d+1$ to be precise) of directions for candidate planes that separate two optimal clusters.

\begin{lemma}\label{lem:k2_separator_R^d}
Consider a point set $P \subset \mathbb{R}^d$ and let $Q \subset P$, $|Q| \leq d+1$, be a set of points that defines the minimum enclosing ball of $P$, i.e., $\MEB(Q) = \MEB(P)$.
There exists an optimal 2-\Minsumradius clustering $C_1, C_2$ of $P$ such that for a point $q\in Q$ and a hyperplane $\ell$ orthogonal to $d(q)$, the hyperplane $\ell$ separates $C_1$ from $C_2$.
\end{lemma}

Before proving this lemma, let us first consider the algorithm.
In order to extend Algorithm \ref{alg:k2} so that it works in $\mathbb{R}^d$, we need an algorithm to dynamically maintain the minimum enclosing ball of a set of points in $\mathbb{R}^d$ under insertions and deletions of points. This can be done in sublinear time by a result of Agarwal and Matousek \cite{dynamic91}. Namely, the amortized time per update is $\Oh\left(n^{1-1/(\lceil d/2 \rceil+1)+\varepsilon}\right)$. Since for every candidate direction we need to do at most $n$ updates (as in the planar case), we have the following result.

\begin{theorem}
For a set $P$ of $n$ points in  $\mathbb{R}^d$ an optimal 2-\Minsumradius clustering can be computed in $\Oh\left(n^{2-1/(\lceil d/2 \rceil+1)+\varepsilon}\right)$ time.
\end{theorem}

\begin{proof}
    Let $Q \subset P$ be a set of at most $d+1$ points that define $D$.
    For each point $q \in Q$, there exist $\Oh(n)$ combinatorially different hyperplanes that are orthogonal to $d(q)$ and partition~$P$ into two clusters. For each of these clusters we have to update the minimum enclosing ball in time $T(n)=\Oh\left(n^{1-1/(\lceil d/2 \rceil+1)+\varepsilon}\right)$ to obtain the 2-\Minsumradius cost. Taking the minimum over all points and all hyperplanes, we find an optimal solution in time $(d+1) \cdot n \cdot T(n) = \Oh(n \cdot T(n))$.
\end{proof}

Let us now prove the structural lemma.

\begin{proof}[Proof of Lemma~\ref{lem:k2_separator_R^d}]
Analogous to the $d=2$ case, assume that the optimal 2-\Minsumradius solution consists of two non-empty clusters, as otherwise an optimal cluster is simply $\MEB(P)$.
If $\MEB(P)$ is not an optimal 2-\Minsumradius clustering, then the optimal cost is less than the radius of $\MEB(P)$ and $C_1, C_2$ are both non-empty.
Let $Q \subset P, |Q| \leq d+1$ be a set of points that define the minimum enclosing ball of $P$, i.e., $\MEB(Q) = \MEB(P)$, and let $Q$ lie on the boundary of the minimum enclosing ball of $P$, i.e., $Q \subset \partial \MEB(P)$.
We argue that a hyperplane orthogonal to a diameter defined by a point in $Q$ separates $C_1$ and $C_2$. 
These points can be computed in $\Oh(n)$ time~\cite{DBLP:journals/siamcomp/Megiddo83a,DBLP:conf/birthday/Welzl91}.


Therefore, consider a partition of $Q$ into $Q_1$ and $Q_2$ and assume without loss of generality that $Q_1 \subset C_1$ and $Q_2 \subset C_2$.
Let $S(P)$ be the minimum enclosing hypersphere of $P$, i.e., $S(P) = \partial \MEB(P)$.
Let $S_1 = \MEB(Q_1) \cap S(P)$ and $S_2 = \MEB(Q_2) \cap S(P)$, and let $c$ be the center of $S(P)$.
Let $S^*_1$ and $Q_1^*$ be $S_1$ and $Q_1$ mirrored at $c$ respectively, i.e., it is all the antipodal points of $S_1$ and $Q_1$ respectively.
We now consider two cases.



First, suppose that $S^*_1 \cap Q_2 \neq \emptyset$. Then, we can construct a diameter of $\MEB(P)$ that intersects both $S_1$ and $S_2$.
As there is a point $q \in Q_2$ with antipodal point $q^*$ such that $q \subset S_2$ and $q^* \subset S_1$, 
both extremes of the diameter $d(q)$ are contained within a different cluster. As the sum of the radii is less than $\MEB(P)$, there is a hyperplane orthogonal to this diameter $d(q)$ that separates the two clusters.
The case $S_2 \cap Q^*_1 \neq \emptyset$ is symmetric.


Hence, it remains to consider the case $(S^*_1 \cap Q_2) \cup (S_2 \cap Q^*_1) = \emptyset$.
We show that this leads to a contradiction.
Note that $S^*_1 \cap Q_2 = \emptyset$ and $S_2 \cap Q^*_1 = \emptyset$ implies that $(S^*_1 \cap S_2)$ contains no point of $Q_2 \cup Q^*_1$.
Hence, we can construct a hyperplane $F$ through $c$ (and $S^*_1 \cap S_2$ if it is non-empty) that strictly separates $Q^*_1$ and $Q_2$ (i.e., $F$ also does not contain any point of $Q^*_1 \cup Q_2$).
Note that as $F$ contains $c$ and $Q^*_1$ was mirrored at $c$, we have that $Q_1$ and $Q_2$ are on the same side of $F$ but do not intersect it.
Thus, all points of $Q = Q_1 \cup Q_2$ are strictly contained on some half of $\MEB(P)$.
It follows from Lemma~\ref{lem:no_empty_half} that $\MEB(Q) \neq \MEB(P)$, which is a contradiction.

\end{proof}

\section{Near-quadratic algorithm for 3-\Minsumradius}
\label{sec:k3}

In this section, we prove the following result:

\begin{theorem}\label{thm:k3}
For a set $P$ of $n$ points in the Euclidean plane an optimal 3-\Minsumradius can be computed in expected $\Oh(n^2 \log^2 n \log^2 \log n)$ or worst-case $\Oh(n^2 \log^4n)$ time. 
\end{theorem}

We first give the algorithm in Section~\ref{sub:algo_k3}, and then prove the structural result that this algorithm relies on in Section~\ref{sec:k3_cuts}.

\subsection{Algorithm}\label{sub:algo_k3}

Our near-quadratic algorithm for the 3-\Minsumradius problem again relies on the structural insight that there are only linearly many cuts that need to be considered in order to find one that separates one cluster from the two other clusters.
For the separated cluster we can then simply compute the minimum enclosing disk, while for the other two clusters we use our near-linear time algorithm for 2-\Minsumradius.
Hence, we obtain a near-quadratic time algorithm.

The following is the main structural lemma that our algorithm relies on:
\begin{lemma} \label{lem:k3_separator}
Consider a point set $P$ in the Euclidean plane and let $p_1, p_2, p_3$ be three points in $P$ that define the minimum enclosing disk of $P$ (with potentially $p_2 = p_3$).
There exists an optimal 3-\Minsumradius solution $C_1, C_2, C_3$  to $P$ such that for a point $q\in\{p_1,p_2,p_3\}$ and a line $\ell$ orthogonal to $d(q)$, the line $\ell$ separates the cluster containing $q$ from the other two clusters.
\end{lemma}
We prove this lemma in Section~\ref{sec:k3_cuts}.
Using this lemma, we can now give our full algorithm, see Algorithm~\ref{alg:k3}. 

\subparagraph{Algorithm description.} For each point $p_i$, $i=1,2,3$, that defines $D$, we consider all combinatorially distinct separators that are orthogonal to the diameter $d(p_i)$, as shown in \cref{fig:2pointsdefine3clusters} for the case with only one distinct possible diameter. For each such separator, let $A\subset P$ be the set induced by the separator that contains $p_i$, and let $B = P \setminus A$. We compute the minimum enclosing disk of $A$ and the 2-\Minsumradius of $B$ using the algorithm of Theorem~\ref{thm:k2}. Finally, we compare this solution to the best solution found so far, and update it when necessary. Lemma~\ref{lem:k3_separator} implies that the algorithm will find an optimal clustering this way. As there are only linearly many combinatorially distinct separators orthogonal to a fixed line, and a solution for $k=2$ can be computed in expected $\Oh(n\log^2 n \log^2 \log n)$ time, the expected running time of the algorithm is $\Oh(n^2 \log^2 n \log^2 \log n)$. The worst-case running time follows by using the deterministic algorithm to maintain the minimum enclosing disk~\cite{k2quadratic}.

\begin{algorithm}[t]

\begin{algorithmic}[1]
\Procedure{3-\Minsumradius}{$P$}
\State $D \gets $ minimum enclosing disk of $P$
\State $p_1, p_2, p_3$ $\gets$ points defining $D$
\State $S_{\text{OPT}} \gets $ 2-\Minsumradius{($P$)}
\For{each diameter $d(p_i)$ of $p_i$}
    \For{each separator $s$ orthogonal to $d(p_i)$ through a $p \in P$}
		\State $A,B \gets$ sets induced by separator $s$, such that $p_i \in A$
		\State $S \gets (\Call{MED}{A}, \Call{2-\Minsumradius}{B})$
		\If{$S$ is better than $S_{\text{OPT}}$}
            \State $S_{\text{OPT}} \gets S$
		\EndIf
	\EndFor
\EndFor
\State \Return $S_{\text{OPT}}$
\EndProcedure
\end{algorithmic}

\caption{Near-quadratic time algorithm for 3-\Minsumradius.}
\label{alg:k3}
\end{algorithm}

\subsection{Linear number of cuts} \label{sec:k3_cuts}

In this subsection, we prove the correctness of our algorithm. That is, we prove Lemma~\ref{lem:k3_separator} and thus show that it suffices to check a linear number of candidate separators, one of which separates one cluster from the two others. 

\begin{proof}[Proof of Lemma~\ref{lem:k3_separator}]

\begin{figure}
     \centering
     \includegraphics[page=3]{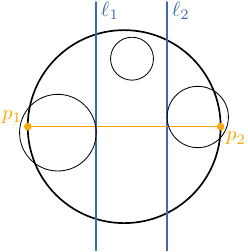}
     \caption{Visualization of the proof for the existence of a separator that is perpendicular to the diameter of $p_1$ if the minimum enclosing disk is defined by two points.}
     \label{fig:2pointsdefine3clusters}
 \end{figure}

 In this proof, we define $\bar C_i=\med(C_i)$. We again assume that $r(C_1)+r(C_2)+r(C_3) < r(D)$, as otherwise the solution with a single cluster trivially satisfies our claim.
 We consider the following cases:

 \subparagraph{Case 1:} \textit{Two points define $D$.}
 Let $p_1$ and $p_2$ be the two points that define $D$. Then $p_1$ and $p_2$ are in different clusters, say $p_1 \in C_1$ and $p_2 \in C_2$.
 Let $\ell_i$ be the line perpendicular to $d(p_1)$ and tangent to $\bar C_i$, furthest in direction $\overrightarrow{c-p_i}$, for $i\in\{1,2\}$.
We claim that the lines have the following properties:
 (i) $\ell_1\cap\bar C_2=\emptyset$, (ii) $\ell_2\cap\bar C_1=\emptyset$, and (iii) $\ell_1\cap\bar C_3=\emptyset$ or $\ell_2\cap\bar C_3=\emptyset$.
 If case (i) or (ii) is false, the projection of the union $\bar C_1\cup \bar C_2$ on the diameter $d(p_1)$ covers the entire diameter, so the sum of radii is at least $r(D)$.
 If case (iii) is false, the same holds for the union $\bar C_1\cup \bar C_2\cup \bar C_3$ of all the disks.
  If $\ell_1\cap\bar C_3=\emptyset$ then, depending on which side of $\ell_1$ contains $C_3$, $\ell_1$ either separates $C_2$ from $C_1$ and $C_3$ or $C_1$ from $C_2$ and $C_3$. Otherwise, if $\ell_2\cap\bar C_3=\emptyset$, then $\ell_2$ separates $C_2$ from $C_1$ and $C_3$ or $C_1$ from $C_2$ and $C_3$.

 \subparagraph{Case 2: } \textit{Three points define $D$, two of which are in the same cluster.}
Suppose $D$ is defined by three points $p_1, p_2, p_3$, where $p_1, p_2 \in C_1$.
 Then we add an artificial point $p_3^*$ diametrically opposing $p_3$ on $D$.
 This reduces this case to Case 1.

 \subparagraph{Case 3: } \textit{Three points define $D$, all of which are in different clusters.}
 \begin{figure}
     \centering
     \includegraphics{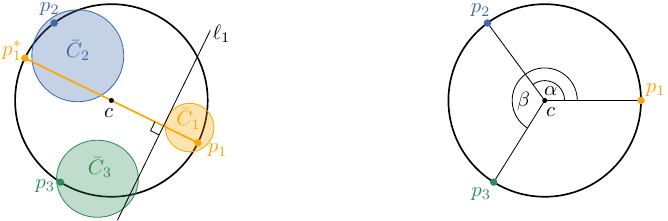}
     \caption{Illustration for the notation used in Case 3.
     Note that to the left, $\bar C_1$ is blocked by $\bar C_3$.}
     \label{fig:notation_k3}
 \end{figure}
 Now we consider the final case, where $p_1\in C_1,p_2\in C_2, p_3\in C_3$. This is also the most technically involved case.
 For simplicity, we assume that $D$ is a unit disk.
 Define lines $\ell_1,\ell_2,\ell_3$ as follows: for $i\in\{1,2,3\}$, $\ell_i$ is the line perpendicular to $d(p_i)$ that is tangent to $\bar C_i$ furthest in the direction $\overrightarrow{c-p_i}$; see Figure~\ref{fig:notation_k3} for an illustration of our notation.
 If a line $\ell_i$ separates $C_i$ from the other clusters, we are done.
 We say that a disk $\bar C_j$ \emph{blocks} a disk $\bar C_i$ if $\bar C_j\cap L_i\neq\emptyset$, where $L_i$ is the half-plane bounded by $\ell_i$ that contains $\bar C_i$. 
 If $\ell_i$ is not a separator, it has to be the case that $L_i$ intersects one of the other two disks.
 We show that if each disk $\bar C_i$ is blocked by another disk, then $r(C_1)+r(C_2)+r(C_3)\geq 1$, which is a contradiction.

Without loss of generality, assume that the center of $D$ is $c=(0,0)$, that $p_1=(1,0)$ and that the points appear in the order $p_1,p_2,p_3$ around $\partial D$ in counterclockwise direction.
Let $\alpha$ be the angle of $p_2$ and $\beta$ be the angle of $p_3$.
Note that by our assumptions, we have $\alpha\in(0,\pi]$ and $\beta\in[\pi,\pi+\alpha]$.
There are two subcases: (3a) there is a disk that blocks both of the other disks, and (3b) the disks block each other in a cyclic pattern.

\subparagraph{Case 3a:}
Assume without loss of generality that $\bar C_2$ blocks $\bar C_1$ and $\bar C_3$.
The reductions described in the following are shown in Figure~\ref{fig:2}.
For $i\in\{1,3\}$, consider the disk $\bar C_i'$ that is tangent to $D$ at $p_i$ and has the tangent $\ell_i$.
If not $\bar C_i=\bar C_i'$, then we can replace $\bar C_i$ with $\bar C_i'$, and this will decrease $r(C_1)+r(C_2)+r(C_3)$ while maintaining that $\bar C_2$ blocks both of the other disks.
After the replacement, the disks no longer contain all input points $P$, but we just do it to obtain a lower bound on the sum of radii of the original disks.
Hence, we assume that $\bar C_i$ is tangent to $D$ at $p_i$.
Furthermore, allowing the tangent $\ell_i$ to change as well, we can assume that $\bar C_i$ is the smallest disk tangent to $D$ at $p_i$ such that $\bar C_2$ blocks $\bar C_i$, which means that $\ell_i$ is a separating common tangent of $\bar C_i$ and $\bar C_2$.

\begin{figure}
\centering
\includegraphics[page=3]{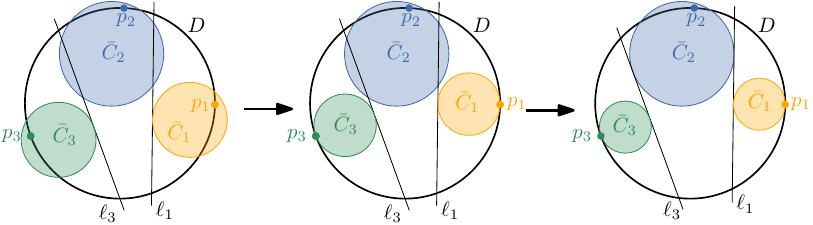}
\caption{
The reduction used in case 3a to a situation where $\bar C_1$ and $\bar C_3$ are as small as possible given the points $p_1,p_3$.
}
\label{fig:2}
\end{figure}

\begin{figure}
\centering
\includegraphics[page=4]{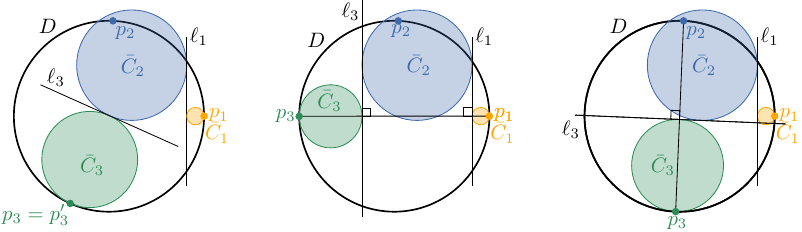}
\caption{Situations in case 3a.
Left: Placement of $p_3$ where $r(C_3)$ is maximum.
Middle: Minimum when $p_1$ and $p_3$ are diametral.
Right: Minimum when $p_2$ and $p_3$ are diametral.}
\label{fig:4}
\end{figure}

We now fix $\bar C_1$ and $\bar C_2$ and consider the angle $\beta\in[\pi,\pi+\alpha]$ of $p_3$ that causes $r(C_3)$ to be as small as possible.
For a value $\beta$, let $\rho(\beta)=r(C_3)$ be the diameter of the resulting disk $\bar C_3$.
Let $p'_3$ be point on the boundary of $D$ on the same diameter as the center of $\bar C_2$ but on the opposite site of the center of $D$.
The radius $r(C_3)$ is maximum when $p_3=p'_3$, in which case the two disks $\bar C_2$ and $\bar C_3$ touch; see Figure~\ref{fig:4} (left).
As $p_3$ moves away from $p'_3$ in either direction, $\bar C_3$ shrinks monotonically.
It follows that $\rho$ has no local minima except at one or both of the endpoints $\beta\in\{\pi,\pi+\alpha\}$.
When $\beta=\pi$, the segment $p_1p_3$ is a diameter of $D$ and the projection of the disks $\bar C_1,\bar C_2,\bar C_3$ on $p_1p_3$ covers the entire diameter, so $r(C_1)+r(C_2)+r(C_3)\geq 1$; see Figure~\ref{fig:4} (middle).
When $\beta=\pi+\alpha$, the segment $p_2p_3$ is a diameter of $D$ and the projection of the disks $\bar C_2,\bar C_3$ on $p_2p_3$ covers the entire diameter, so $r(C_1)+r(C_2)+r(C_3)\geq r(C_2)+r(C_3)\geq 1$; see Figure~\ref{fig:4} (right).

\subparagraph{Case 3b:}
Without loss of generality, assume that $\bar C_i$ blocks $\bar C_{i+1}$ for each $i\in\{1,2,3\}$ (with indices taken modulo $3$).
Let $h_i$ be the tangent to $D$ in $p_i$ and $k_i$ be the line through $p_{i-1}$ parallel to $h_i$; see Figure~\ref{fig:3}.
Then the three lines $h_i$, $\ell_i$ and $k_i$ are all parallel.

\begin{figure}[h]
\centering
\includegraphics[page=2]{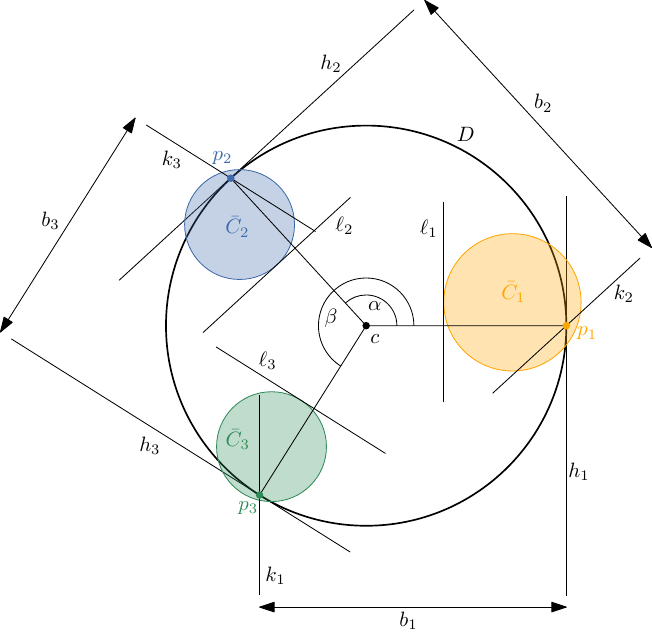}
\caption{The setting in case 3b.}
\label{fig:3}
\end{figure}

Since $\bar C_i$ blocks $\bar C_{i+1}$, the sum $2r(C_i)+2r(C_{i+1})$ is at least the distance between the lines $h_{i+1}$ and $k_{i+1}$.
Define $b_{i+1}$ to be this distance.
We then have $4(r(C_1)+r(C_2)+r(C_3))\geq b_1+b_2+b_3$, and we aim at showing $b_1+b_2+b_3\geq 4$.

We first fix $p_2$ and analyze the placement of $p_3$ that results in the smallest value of $b_3+b_1$.
In other words, we fix $\alpha$ and consider the value of $\beta\in[\pi,\pi+\alpha]$ as a free variable.
We observe that $b_3=1-\cos(\beta-\alpha)$ and $b_1=1-\cos(2\pi-\beta)=1-\cos\beta$.
Hence,
\[
b_3+b_1=2-\cos(\beta-\alpha)-\cos\beta=
2-\cos((\beta-\alpha/2)-\alpha/2)-\cos((\beta-\alpha/2)+\alpha/2).
\]
The sum formula $\cos(x+y)=\cos x\cos y-\sin x\sin y$ gives
\[
b_3+b_1=2-2\cos(\alpha/2)\cos(\beta-\alpha/2).
\]
Now it is clear that $b_3+b_1$ is maximum when $\beta=\pi+\alpha/2$ and minimum when $\beta\in\{\pi,\pi+\alpha\}$.
For both of these minimum values, we get
\[
b_3+b_1=2+2\cos^2(\alpha/2).
\]
Note that $b_2=1-\cos\alpha$, so we get
\[
b_1+b_2+b_3=1-\cos\alpha+2+2\cos^2(\alpha/2)=3-\cos\alpha+2\cos^2(\alpha/2).
\]
By the identity $\cos(2x)=2\cos^2 x-1$, we get
$b_1+b_2+b_3=4$, and we have the desired contradiction.
\end{proof}

\section{Conclusion}\label{sec:conclusion}
In this work, we presented two new clustering algorithms for point sets in the plane: one for $2$-\Minsumradius and one for $3$-\Minsumradius.
We improve the running time for the former problem from near-quadratic to near-linear, and generalize the algorithm to a sub-quadratic algorithm for any constant dimension. For the latter problem we give the first non-naive algorithm, matching the running time of the previously best algorithm for $2$-\Minsumradius.
At the heart of our algorithms are new structural insights on separators of optimal solutions, which might be of independent interest.
We showed that there are only linearly many separators that need to be considered for $k=2$ as well as $k=3$.
Hence, it is conceivable that there also exists a near-linear time algorithm for $3$-\Minsumradius.

An interesting direction for future work is to understand whether any of our structural insights can be pushed further to obtain a significantly improved running time for general $k$-\Minsumradius clustering.
Another natural future line of work is to extend our algorithm for $3$-\Minsumradius to higher dimensions, which could be another step towards an efficient algorithm for $k$-\Minsumradius in higher dimensions.
To understand the inherent hardness of the problems considered in this work, it would be of interest to show fine-grained lower bounds for some combinations of the ambient dimension and $k$, similar to~\cite{DBLP:conf/icalp/ChanHY23}.

We may also consider the \emph{rectilinear} $k$-\Minsumradius problem, where we aim to cover a set of points using $k$ axis-parallel squares of minimum total side length.
For $k=2$, this problem has a straight-forward $\Oh(n\log n)$ time algorithm, which is optimal in the algebraic computation tree model by a reduction from the Maximum Gap problem~\cite{DBLP:journals/algorithmica/LeeW86}.
For $k=3$, it is easy to see that there must be a vertical or horizontal line that separates one square from the other two in an optimal solution, so we obtain an $\Oh(n^2\log n)$ time algorithm.
It is interesting whether a faster algorithm exists.
Again, the analogous $k$-\Center problem of covering a set of points with $k$ squares of minimum maximum size has received much more attention, and~$\Oh(n)$ time algorithms have been described for up to three clusters~\cite{drezner1987rectangular,DBLP:journals/comgeo/Hoffmann05,DBLP:conf/compgeom/SharirW96} and $\Oh(n\log n)$ time algorithms for up to five clusters~\cite{DBLP:conf/issac/Nussbaum97,DBLP:conf/esa/Segal97,DBLP:conf/compgeom/SharirW96}.

\bibliography{main}

\end{document}